\documentclass[preprint,12pt]{elsarticle}

\usepackage{amssymb}

\usepackage{graphicx}
\usepackage{amsmath}
\usepackage{amsthm}
\usepackage{amssymb}
\usepackage{latexsym}
\usepackage{graphicx}
\usepackage{tikz}
\usetikzlibrary{matrix,arrows,decorations.pathmorphing}

\newcommand{\CR}{{\rm CR}}

\newtheorem{theorem}{Theorem}[section]
\newtheorem{definition}[theorem]{Definition}
\newtheorem{lemma}[theorem]{Lemma}

\newtheorem{proposition}[theorem]{Proposition}
\newtheorem{corollary}[theorem]{Corollary}
\newtheorem{algorithm}[theorem]{Algorithm}

\journal{Information Science}

\begin{document}

\begin{frontmatter}



\title{Capacity Factors of a Point-to-point Network}


\author[cs, iipl]{Yuan Li}
\ead{07300720173@fudan.edu.cn}
\author[mt]{Yue Zhao}
\ead{07300720233@fudan.edu.cn}
\author[cs, iipl]{Haibin Kan\corref{cor1}}
\ead{hbkan@fudan.edu.com}
\address[cs]{School of Computer Science, Fudan University, Shanghai 200433, China}
\address[mt]{School of Mathematical Science, Fudan University, Shanghai 200433, China}
\address[iipl]{Shanghai Key Lab of Intelligent Information Processing, Fudan University, Shanghai 200433, China}

\cortext[cor1]{Corresponding author}

\address{}

\begin{abstract}
In this paper, we investigate some properties on capacity factors,
which were proposed to investigate the link failure problem from network coding.
A capacity factor (CF) of a network is an edge set, deleting which will cause the
maximum flow to decrease while deleting any proper subset will not. Generally,
a $k$-CF is a minimal (not minimum) edge set which will cause the network maximum flow decrease by $k$.

Under point to point acyclic scenario, we characterize all the edges which are
contained in some CF, and propose an efficient algorithm to classify.
 And we show that all edges on some $s$-$t$ path in an acyclic point-to-point acyclic network are
 contained in some $2$-CF.
We also study some other properties of CF of point to point network, and a simple relationship with CF in multicast
network.

On the other hand, some
computational hardness results relating to capacity factors are
obtained. We prove that deciding whether there is a capacity factor of a cyclic network
with size not less a given number is NP-complete, and the time
complexity of calculating the capacity rank is lowered bounded by
solving the maximal flow. Besides that, we propose the analogous definition of
CF on vertices and show it captures edge capacity factors as a special case.
\end{abstract}


\end{frontmatter}


\section{Introduction}
Reliability is a critical theme in the topology design of
communication networks. In traditional combinatorial network theory,
we use the concept edge-connectivity to evaluate the reliability of
a network in case of edge failures \cite{Boe86}. A lot of intense
studies have been focused on the connectivity of a network and
extending concepts such as super connectivity and conditional
connectivity \cite{Boe86}, \cite{EH98}, \cite{Har83}.

The more recent study of network coding shows that \cite{ACLY00} by
coding at internal vertices, one can achieve the optimal capacity of
a multicast network, which is upper bounded by the maximum flow (minimum cut) of the network. However, this optimal capacity can not
be achieved by traditional routing scheme.

Here comes a question: for a network coding-based network, is it
still appropriate to evaluate its reliability by traditional
concepts like edge-connectivity as mentioned above? For traditional
networks where only routing schemes are adopted, the communication
will not be ruptured in case of edge failures as long as at least
one path from source to the sink vertex still exists. However, in
network coding-based network, the communication will be degraded
even if the failures of an edge set may reduce the number of
disjoint paths between source and sink vertices for the network
capacity is decreased.

Koetter et al. \cite{KM03} first mentioned the edge failure problems
in network coding-based networks. Cai and Fan \cite{CF07} formally
proposed the concept of capacity factor and capacity rank. The
capacity rank characterizes the criticality of a link for the
network communication. When there is no capacity factor containing
an edge, the capacity rank of this edge is defined as $\infty$.
Recently, we notice some related work in \cite{CMVZ12} about $k$-Route Cut,
which is the minimum number of edges to let the connectivity of every pair of source and sink
falls below $k$. They not only propose some approximation algorithms, but
also prove some computational hardness results.
In fact, the generalized definition of $k$-CF captures some definitions such
as $k$-route cute when the connectivity is defined as maximal number of
edge-disjoint paths.


In \cite{CF07}, the authors proposed an open problem which
is deciding the capacity rank of a given edge. In this paper, we
obtain an equivalent condition that deciding whether $\CR(e)=\infty$. By this result, it is easy
to develop an algorithm in $O(V^3)$ to determine whether the capacity rank of a given edge is finite or
infinite, which partially answers the open problem in \cite{CF07}.
Although neither can we find an efficient algorithm to compute the
capacity rank of a general network, nor can we prove the problem is
NP-hard, we obtain some computational hardness results relating to
it. For example, we prove that deciding whether there is a capacity
factor with size not less than a given number is NP-complete, and
show the time complexity of computing the capacity rank is lower
bounded by that of solving maximal flow.

Even though there is no benefit
of network coding on single-source single-sink network, our results
are mainly focused on single-source single-sink scenario. There are
mainly two reasons: the study on point-to-point network may bring insights
into multicast scenario;
the CF is a natural concept, which might be interesting on its own right.

This paper is organized as follows. In section
\uppercase\expandafter{\romannumeral2} , we review some basic
definitions, notations and related results. In section
\uppercase\expandafter{\romannumeral3}, we investigate some
properties of capacity factors, including both point-to-point scenario and multicast
scenario.
In section \uppercase\expandafter{\romannumeral4}, we propose an
algorithm to calculate the $D$-set and $H$-set
with its correctness proved and time complexity analyzed.
In section
\uppercase\expandafter{\romannumeral5}, we present some computational hardness results relating to
capacity factors.

\section{Preliminaries}
In this section, we review some basic definitions, notations and
results, which will be used in the sequel.

A communication network is a collection of directed links connecting
transmitters, switches, and receivers. It is often represented by a
4-tuple $\mathcal{N}=(V,E,S,T)$, where $V$ is the vertex (node) set,
$E$ the edge (link) set, $S$ the source vertex set and $T$ the sink
vertex set. A communication network $\mathcal{N}$ is called a
point-to-point communication network if $|S|=|T|=1$, denoted by
$\mathcal{N}=(V,E,s,t)$, where $s$ is the source vertex and $t$ the
sink vertex.

Without loss of generality, we may assume that all links in a
network have the same capacity, $1$ bit per transmission slot. For
$u, v \in V$, denote by $\langle u,v \rangle$ the edge from $u$ to
$v$. If there are $k$ edges from $u$ to $v$, we denote by $\langle
u, v \rangle_k$ the set consisting of edges from $u$ to $v$, or
simply denote it by $\langle u,v \rangle$ when there is no
ambiguity. For an edge $e= \langle u,v \rangle$, $u$ is called the
tail of $e$ and denoted by $tail(e)$, and $v$ is called the head of
$e$ and denoted by $head(e)$.

If $F \subseteq E$, denote by $\mathcal{N}\backslash F$ the network
obtained by deleting edges in $F$ from $\mathcal{N}$. If $V'
\subseteq V$, denote by $\mathcal{N}(V')$ the network consisting of
vertices in $V'$ and the edges among $V'$ of $\mathcal{N}$, calling
the vertex-induced network of $\mathcal{N}$ by $V'$. For $V_1, V_2
\subseteq V$, denote by $[V_1, V_2]$ the set consisting of all links
with tails in $V_1$ and heads in $V_2$. For a network
$\mathcal{N}=(V,E,S,T)$, an $S$-$T$ cut of $\mathcal{N}$ is $[V_1,
\overline{V_1}]$ such that $V_1$ is a subset of $V$ containing all
vertices in $S$ but not containing any vertex in $T$. A minimal
$S$-$T$ cut of $\mathcal{N}$ is an $S$-$T$ cut with the minimal
size, denoted by $C_{\mathcal{N}}(S,T)$.

It is well known that, for a point-to-point network
$\mathcal{N}=(V,E,s,t)$, the maximal flow from $s$ to $t$ is equal
to the minimal $s$-$t$ cut of $\mathcal{N}$ and a feasible flow is a
maximal flow if and only if there is no augmenting path in the
corresponding residual network (\emph{Max-flow Min-cut Theorem
\cite{CLRS01}, \cite{Wes01}}). If each link in $\mathcal{N}$ has
unit capacity, then the maximal flow $f$ of $\mathcal{N}$
corresponds to $|f|$ edge-disjoint paths from $s$ to $t$ in
$\mathcal{N}$ (\emph{Integrality Theorem \cite{Wes01}}), where $f$
denotes a collection of edge-disjoint paths (a flow) and $|f|$
denotes the number of the paths. Throughout the paper, we always
assume each link has unit capacity.

Let $\mathcal{N}=(V,E,s,t)$ be a point-to-point network. For any
vertex $v \in V$, we can assume that there exists a path from $s$ to
$t$ in $\mathcal{N}$ which passes the vertex $v$. Otherwise, we can
delete the vertex $v$ because $v$ is useless for the communication
between $s$ and $t$ in $\mathcal{N}$. Similarly, for any edge $e \in
E$, we can assume that there exists a path form $s$ to $t$ which
passes the edge $e$.

\begin{definition}\cite{CF07}
Let $\mathcal{N}=(V,E,s,t)$ be a point-to-point network. A nonempty
subset $F$ of $E$ is a capacity factor of $\mathcal{N}$ if and only
if the following two conditions hold:
\begin{enumerate}
\item $C_{\mathcal{N} \backslash F}(s,t) < C_{\mathcal{N}}(s,t)$;
\item $C_{\mathcal{N} \backslash F'}=C_{\mathcal{N}}(s,t)$ for any
proper subset $F' of F$.
\end{enumerate}
$\mathcal{N} \backslash F$ denotes the induced network formed by
deleting $F$ in $\mathcal{N}$.
\end{definition}

By this definition, for a capacity factor $F$, adding any one edge
$e \in F$ in the point-to-point network $\mathcal{N} \backslash F$
will increase the maximal flow. Since adding one edge can increase
the maximal flow by at most $1$, we have $C_{\mathcal{N} \backslash
F}(s,t) = C_{\mathcal{N}}(s,t)-1$.

Generally, we can define $k^\text{th}$ order capacity factor ($k$-CF) of
as follows, where the motivation will be clear in the multicast scenario.

\begin{definition}
Let $\mathcal{N}=(V,E,s,t)$ be a point-to-point network. A nonempty
subset $F$ of $E$ is a $k^\text{th}$ order capacity factor ($k$-CF) of $\mathcal{N}$ if and only
if the following two conditions hold:
\begin{enumerate}
\item $C_{\mathcal{N} \backslash F}(s,t) \le C_{\mathcal{N}}(s,t) - k$;
\item $C_{\mathcal{N} \backslash F'} > C_{\mathcal{N}}(s,t) - k$ for any
proper subset $F' of F$.
\end{enumerate}
$\mathcal{N} \backslash F$ denotes the induced network formed by
deleting $F$ in $\mathcal{N}$.
\end{definition}

Since the network coding capacity of a single-source multi-sink
network $\mathcal{N}=(V, E, s, T)$, where $T=(t_1, \ldots, t_m)$, is
upper bounded by the minimal of the maximal flow from $s$ to $t_i$ \cite{ACLY00}, i.e.,
the capacity region is $$(C_\mathcal{N}(s, t_1), C_\mathcal{N}(s, t_2), \ldots, C_\mathcal{N}(s, t_m)).$$
Therefore, we have the following definition of $k$-CF on a
multicast network.

\begin{definition} $\mathcal{N}=(V,E,s,T)$ is a multicast network and $T=\{t_1, \ldots, t_m\}$,
edge set $F$ is a $(k_1, k_2, \ldots, k_m)^{\text{th}}$ order capacity factor ($\overrightarrow{k}$-CF) of $\mathcal{N}$ if and only if:
\begin{itemize}
\item[(1)] For all $1 \le i\le m, C_{\mathcal{N}\backslash F}(s, t_i) \leq C_\mathcal{N}(s, t_i)-k_i$;
\item[(2)] For any $F' \subsetneq F$, there exists $i \in \{1, \ldots, m\}$, such that $C_{\mathcal{N}\backslash F'}(s,
t_i) > C_\mathcal{N}(s, t_i)-k_i$.
\end{itemize}
When $\overrightarrow{k} \neq 0$ and $k_i \le 1$, we say $F$ is a CF of multicast network $\mathcal{N}$.
\end{definition}

One could easily generalize the above definition to multi-source multi-sink network. This generalized
definition captures the $k$-route cut problem of edge-connectivity version, i.e., $k$-route cut is
the minimum $\overrightarrow{k}$-CF where $\overrightarrow{k}=(f_1-k+1, \ldots, f_t-k+1)$.


Following is the definition of D-set and H-set, which is simply the
edge union of all CFs and all the remaining ones.

\begin{definition}\cite{CF07}
Let $\mathcal{N}=(V,E,s,T)$ be a point-to-point or multicast communication
network. The collection of all its capacity factors $
\mathcal{D}=\{F_1, F_2, \ldots, F_r\}$ is called the
capacity factor set of $\mathcal{N}$. While
$D=\bigcup^{r}_{i=1}{F_i}$ is called the $D$-set
of $\mathcal{N}$ and $H=E \backslash D$ is
called the $H$-set of $\mathcal{N}$.
\end{definition}

By the definition, it is not difficult to see that $C_{\mathcal{N}
\backslash H}(s,t)=C_{\mathcal{N}}(s,t)$. Thus, the edge
set of a point-to-point network can be decomposed into two disjoint
parts, namely $D$-set and $H$-set, which
represent the relatively important links and the unimportant links.
 However, this classification is a little
rough. The next definition gives a concept characterizing the
criticality of a link more precisely.

\begin{definition}\cite{CF07} Let $\mathcal{N}=(V,E,s,t)$ be a point-to-point
network. The capacity rank of a edge $e\in E$ is the minimum size of
the capacity factors containing $e$, denoted by
$CR_{\mathcal{N}}(e)$ or $CR(e)$ when there is no ambiguity. If
there is no capacity factor containing $e$, we define
$CR(e)=\infty$.
\end{definition}

The links with smaller capacity ranks are of higher criticality. How
to calculate the capacity rank of a given edge? As far as we know, the
problem is still open. A direct idea
is to find all the capacity factors and then decide the capacity
rank of each edge. The following example shows it's impractical
since the number of capacity factors may grow exponentially with the
size of network. 

\begin{figure}[h]
\begin{center}
\begin{tikzpicture}
[inner sep=1mm,
cir/.style={circle,draw,thick},
pre/.style={->,>=stealth',semithick},
post/.style={dotted,->,>=stealth',semithick}]
\node (s) at (-4,0) [cir,label=left:$s$]{};
\node (v1) at (-2,0) [cir,label=above:$s^{'}$]{};
\node (v2) at (0,2) [cir]{};
\node (v3) at (0,1) [cir]{};
\node (v4) at (0,0) [cir]{};
\node (v5) at (0,-1) [inner sep=0.3mm,cir]{};
\node (v6) at (0,-2) [cir]{};
\node (v7) at (2,2) [cir]{};
\node (v8) at (2,1) [cir]{};
\node (v9) at (2,0) [cir]{};
\node (v10) at (2,-1) [inner sep=0.3mm,cir]{};
\node (v11) at (2,-2) [cir]{};
\node (t) at (4,0) [cir,label=right:$t$]{};
\draw [pre] (s) to node[above]{$e_0$} (v1);
\draw [pre] (v1) to [bend left=30] node[above]{$e_1$} (v2);
\draw [pre] (v2) to node[above]{$e_2$} (v7);
\draw [pre] (v7) to [bend left=30] node[above]{$e_3$} (t);
\draw [pre] (v1) to [bend left=20] node[above]{$e_4$} (v3);
\draw [pre] (v3) to node[above]{$e_5$} (v8);
\draw [pre] (v8) to [bend left=20] node[above]{$e_6$} (t);
\draw [pre] (v1) to node[above]{$e_7$} (v4);
\draw [pre] (v4) to node[above]{$e_8$} (v9);
\draw [pre] (v9) to node[above]{$e_9$} (t);
\draw [post] (v1) to [bend right=20] (v5);
\draw [post] (v5) to (v10);
\draw [post] (v10) to [bend right=20] (t);
\draw [pre] (v1) to [bend right=30] node[right]{$e_{3n-2}$} (v6);
\draw [pre] (v6) to node[above]{$e_{3n-1}$} (v11);
\draw [pre] (v11) to [bend right=30] node[left]{$e_{3n}$} (t);
\end{tikzpicture}
\caption{Network with exponential number of 1-CFs}
\label{fig:unit_network}
\end{center}
\end{figure}
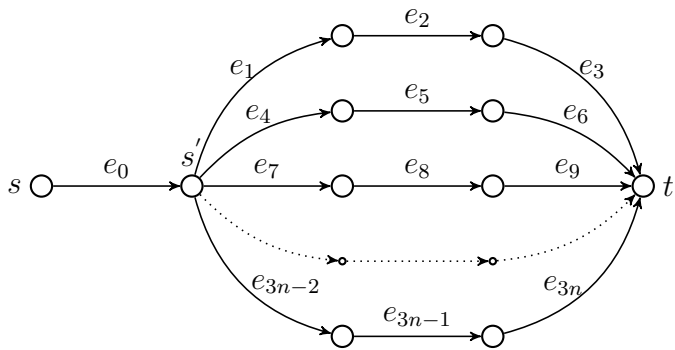

Consider the network $\mathcal{N}=(V,E,s,t)$ shown in Figure \ref{fig:unit_network}, where $|V|=2n+3, |E|=3n+1$. The network can
be decomposed into $n$ internally disjoint paths from $s'$ to $t$
and an individual edge $e_0$. The maximal flow of $\mathcal{N}$ is
$1$. Let $F = \{e_{i_1}, e_{i_2}, \ldots, e_{i_n}\}$, $3(k-1)+1 \leq
i_k \leq 3k$ for $k = 1, 2, \ldots, n$. Due to the simple structure
of the network, it is easy to see that $F$ is a capacity factor of
$\mathcal{N}$ and all capacity factors besides $\{e_0 \}$ can be
written in the form of $F$. Therefore, the total number of capacity factor
of $\mathcal{N}$ is $3^n+1$, which grows exponentially with
$|V|+|E|$.

\section{Some Properties of Capacity Factor}

The following lemma is very useful, which guarantees the existence
of some $k$-CF containing a specific edge $e$ when some conditions are satisfied.

\begin{lemma}
\label{pros:factor_exist}
Let $\mathcal{N}=(V,E,s,t)$ be a point-to-point network.
For an edge $e\in E$ and an integer $1\leq k \leq C_{\mathcal{N}}(s,t)$, if
there exists a edge set $F' \subseteq E$ containing $e$ such that $C_{\mathcal{N}\setminus F^{'}}(s,t)=f-k,\ C_{\mathcal{N}\setminus (F^{'}\setminus e)}(s,t)>f-k$, then there exists a $k$-CF containing $e$.
\end{lemma}
\begin{proof}
Let $\mathcal{F}_{e}^{k}=\{ F'\subseteq E \mid e\in F',\ C_{\mathcal{N}\setminus F'}(s,t)=f-k,\ C_{\mathcal{N}\setminus (F'\setminus e)}(s,t)>f-k \}$.
By condition, $\mathcal{F}_{e}^{k}$ is not empty. Hence we can find a $F\in \mathcal{F}_{e}^{k}$ with minimal cardinality. In fact, $F$ is what we want. Firstly, we have $C_{\mathcal{N}\setminus F}(s,t)=f-k$. Secondly, $\forall\ e' \in F$, if $e'=e$ then we already have $C_{\mathcal{N}\setminus (F^\setminus e)}(s,t)>f-k$ otherwise if $C_{\mathcal{N}\setminus (F\setminus e')}(s,t)=f-k$ then $F\setminus e'$ still belongs to $\mathcal{F}_{e}^{k}$ which contradicts with that $F$ has the minimal cardinality. So $C_{\mathcal{N}\setminus (F\setminus e')}(s,t)>f-k$ for all $e' \in F$, which implies $F$ is a $k$-CF.
\end{proof}

The following proposition only holds for acyclic network, and therefore
all properties depending on it only holds for acyclic network.

\begin{proposition}
\label{pros:decompose} Let $\mathcal{N}=(V,E,s,t)$ be an acyclic network. If $\mathcal{N}$ can be decomposed into
$C_{\mathcal{N}}(s,t)$ edge-disjoint paths, then for any $e \in E$,
we have $C_{\mathcal{N} \backslash e}(s,t) =
C_{\mathcal{N}}(s,t)-1$.
\end{proposition}
\begin{proof} Let $\mathcal{N'}=\mathcal{N} \backslash e$,
$m=C_{\mathcal{N}}(s,t)$. Denote the $m$ edge-disjoint paths by
$p_1$, $p_2$, $\dots$, $p_m$. Since $\mathcal{N}$ can be decomposed into $C_{\mathcal{N}}(s,t)$ disjoint paths, $e$ must be on
one of the path.
Without loss of generality, we assume
$e$ is on $p_m$ and $p_m=(u_1, u_2, \ldots, u_k, u_{k+1} \ldots,
u_l)$, where $\langle u_k, u_{k+1} \rangle = e$ and $1 \leq k \leq l
- 1$.

Clearly, there is a feasible flow $f$ on $\mathcal{N'}$, which is
consisting of $m-1$ edge-disjoint paths $p_1, p_2, \ldots, p_{m-1}$
and hence $C_{\mathcal{N'}}(s,t) \geq m-1$. Recalling \emph{Max-flow
min-cut theorem}, we know that a flow $f$ is a maximal flow if and
only if there is no augmenting paths on the residual network. So it
is sufficient to show the residual network $\mathcal{N'}_f$ has no
augmenting path.

Since $\mathcal{N}$ is acyclic, we assign each vertex an integer label by topological order, such
that $\langle u, v \rangle \in E$ implies $L(u) < L(v)$.

According to the edge direction in network $\mathcal{N}$, all the
edges in the residual network $\mathcal{N'}$ can partitioned into
two parts, forward edges and reversal edges. Consider all the
forward edges in $\mathcal{N'}_f$, which can be viewed as the union
of two paths, $(s = u_1, \ldots, u_k) \cup (u_{k+1}, \ldots, u_l = t)$. By the label properties of vertices, for a
forward edge $(u, v)$ and a reversal edge $(u', v')$, we have
$L(u) < L(v)$ and $L(u') > L(v')$ respectively.
Since $u_1 < u_2 < \dots <u_l $, there is no reversal edge with head
in $\{u_{k+1}, \ldots, u_l\}$ and tail in $\{u_1, \ldots, u_{k}\}$.
Therefore, $u_1$($=s$) and $u_l$($=t$) are disconnected and there is
no augmenting path in the residual network $\mathcal{N'}_f$. This
completes our proof.
\end{proof}

The following result shows that any $k$-CF is contained in some $(k+1)$-CF, assuming the
maximal flow of network is greater than $k$ of course.

\begin{proposition}
\label{pros:cf_contain}
Let $\mathcal{N}=(V,E,s,t)$ be an acyclic network, and assume $F$ is a $k$-CF of $\mathcal{N}$, where $k<C_{\mathcal{N}}(s,t)$, there exists a $(k+1)$-CF $F^{'}$ such that $F\subseteq F^{'}$.
\end{proposition}
\begin{proof} Let $f=C_{\mathcal{N}}(s,t)$. Assume $e$ is a cut-edge of the network $\mathcal{N}\setminus F$. Let $F^{'}=F\cup \{e \}$, then $F^{'}$ is a $(k+1)$-CF of $\mathcal{N}$. Because $C_{\mathcal{N}\setminus F^{'}}=f-(k+1)$ and $\forall\ \tilde{F}\subsetneq F^{'}$, $C_{\mathcal{N} \setminus \tilde{F}}(s,t)>f-k$ if $e\notin \tilde{F}$ otherwise $C_{\mathcal{N}\setminus \tilde{F}}(s,t) \geq f-k$ for deleting one edge at most diminishes one flow.
\end{proof}

The following theorem  characterizes an edge which is contained
in some $k$-CF of a point-to-point acyclic network.

\begin{theorem}
\label{thm:main}
Let $\mathcal{N}=(V,E,s,t)$ be an acyclic point-to-point network, and integer $1 \le k \le C_{\mathcal{N}}(s,t)$.
For any edge $e\in E$, there is a $k$-CF $F$ containing $e$ if and only if there exists an $s$-$t$ path $p$ containing $e$ such that $C_{\mathcal{N}\setminus p}(s,t)\geq C_{\mathcal{N}}(s,t)-k$.
\end{theorem}
\begin{proof}
Let $f=C_{\mathcal{N}}(s,t)$. $``\Rightarrow$'': $e \in F$, $F$ is a $k$-CF. Since $C_{\mathcal{N}\setminus F}(s,t)=f-k$, while adding $e$ to $\mathcal{N}\setminus F$ the maximum flow increase by $1$, we claim there is a path $p$ containing $e$ such that $C_{(\mathcal{N}\setminus F)\setminus p}(s,t)=f-k$ which implies $C_{\mathcal{N}\setminus p}(s,t)\geq f-k$.

$``\Leftarrow$'': By Proposition \ref{pros:cf_contain}, assume $C_{\mathcal{N}\setminus p}(s,t)= f-k$. Since $C_{\mathcal{N}\setminus p}(s,t)= f-k$, there is a feasible flow on $\mathcal{N}\setminus p$, which can be decomposed into $f-k$ paths $p_1,\ p_2\, \ldots,\ p_{f-k}$. Denote by $F=E\setminus (\bigcup\limits_{i=1}^{f-k} p_i \cup p)$, by Proposition \ref{pros:decompose}, $E\setminus (F\cup {e})=(E\setminus F)\setminus e=\bigcup\limits_{i=1}^{f-k}p_i\cup (p\setminus e)$ implies $C_{\mathcal{N}\setminus (F\cup {e})}(s,t)=C_{\bigcup\limits_{i=1}^{f-k}p_i\cup (p\setminus e)}(s,t)=f-k$. And $C_{\mathcal{N}\setminus F}(s,t)=C_{\bigcup\limits_{i=1}^{f-k} p_i \cup p}(s,t)=f-k+1>f-k$. Apply Lemma \ref{pros:factor_exist}, we know there is always a $k$-CF $F$ containing $e$.
\end{proof}

Recall that an edge of a network is either in the
$D$-set is defined as the union of all CFs and $H$-set consists of all the remainings. Taking $k=1$ in the of preceding theorem, we obtain the following
result, which gives an equivalent condition to characterize edges in
the $D$-set and $H$-set.

\begin{corollary}
\label{cor:main} Let $\mathcal{N}=(V,E,s,t)$ be a point-to-point acyclic
network. For any $e \in E$ on some $s$-$t$ path, $e$ is in the $H$-set if and
only if for any $s$-$t$ path containing $e$, $C_{\mathcal{N}
\backslash P}(s,t) \leq C_{\mathcal{N}}(s,t)-2$; and $e$ is in the $D$-set
if and only if there exists some $s$-$t$ path $p$ containing $e$ such that $C_{\mathcal{N}
\backslash P}(s,t) = C_{\mathcal{N}}(s,t)-1$.
\end{corollary}

\begin{figure}[h]
\begin{center}
\begin{tikzpicture}
[inner sep=1mm,
cir/.style={circle,draw,thick},
pre/.style={->,>=stealth',semithick}]
\node (v1) at (0,0) [cir,label=left:$s$]{};
\node (v2) at (2,1) [cir,label=above:$v_1$]{};
\node (v4) at (4,1) [cir,label=above:$v_3$]{};
\node (v3) at (2,-1) [cir,label=below:$v_2$]{};
\node (v5) at (4,-1) [cir,label=below:$v_4$]{};
\node (v6) at (6,0) [cir,label=right:$t$]{};
\draw [pre] (v1) to node[above]{$e_1$} (v2);
\draw [pre] (v2) to node[above]{$e_3$} (v4);
\draw [pre] (v4) to node[above]{$e_6$} (v6);
\draw [pre] (v4) to node[above]{$e_4$} (v3);
\draw [pre] (v1) to node[below]{$e_2$} (v3);
\draw [pre] (v3) to node[below]{$e_5$} (v5);
\draw [pre] (v5) to node[below]{$e_7$} (v6);
\end{tikzpicture}
\caption{Example for D-set and H-set}
\label{fig:posmain}
\end{center}
\end{figure}
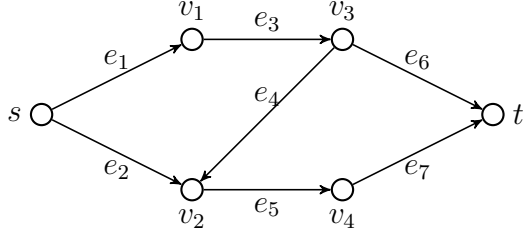

Consider the point-to-point acyclic network $\mathcal{N} = (V,E,s =
v_1,t = v_6)$ shown in Figure \ref{fig:posmain}. There are three
different $s$-$t$ path in total, which are $p_1 = (s, v_1, v_3,
t)$, $p_2 = (s, v_2, v_4, t)$ and $p_3 = (s, v_1, v_3, v_2,
v_4, t)$. It's clear that $C_{\mathcal{N} \backslash p_1}(s,t) =
C_{\mathcal{N} \backslash p_2}(s,t) = 1$, but $C_{\mathcal{N}
\backslash p_3}(s,t) = 0$. Thus $e_4 = \langle v_3, v_2 \rangle$ is
the only edge satisfying the conditions in Corollary \ref{cor:main}.
Therefore $H = \{ e_4 \}$ and $D = E \backslash
H = \{ e_1, e_2, e_3, e_5, e_6, e_7 \}$, which coincides
with the direct computation that $D =
\bigcup_{\textrm{CF } F}{F}= \{ e_1 \} \cup \{ e_2 \}
\cup \{ e_2 \} \cup \{ e_3 \} \cup \{ e_5 \} \cup \{ e_6 \} \cup \{
e_7 \} = \{ e_1, e_2, e_3, e_5, e_6, e_7 \}.$

\vspace{0.3cm}
Even though network $\mathcal{N}$ is assumed to be acyclic in traditional
network coding, we still want to know whether this
characterization holds in a cyclic network. In fact, if an edge
is in $D$-set, then there exists a path $p$ containing it
and satisfying $C_{\mathcal{N}\backslash p}(s, t) =
C_{\mathcal{N}}(s, t)-1.$ It's easy to check that proof procedure of
necessity in Theorem \ref{thm:main} also holds for cyclic network. But the sufficiency proof does not work.
Because Proposition \ref{pros:decompose} does not hold for cyclic
network. The following is a counterexample.

\begin{figure}[h]
\begin{center}
\begin{tikzpicture}
[inner sep=1mm,
cir/.style={circle,draw,thick},
pre/.style={->,>=stealth',semithick}]
\node (v1) at (0,0) [cir,label=left:$s$]{};
\node (v2) at (2,1) [cir,label=above:$v_1$]{};
\node (v3) at (2,-1) [cir,label=below:$v_2$]{};
\node (v6) at (4,0) [cir,label=right:$t$]{};
\draw [pre] (v1) to  (v2);
\draw [pre] (v1) to  (v3);
\draw [pre] (v2) to [bend left=30]   (v3);
\draw [pre] (v3) to [bend left=30]   (v2);
\draw [pre] (v2) to  (v6);
\draw [pre] (v3) to  (v6);
\end{tikzpicture}

\caption{A counterexample for cyclic network}
\label{fig:counterexample}
\end{center}
\end{figure}
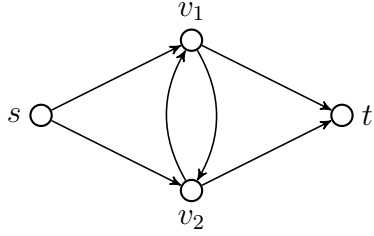

Consider a cyclic network $\mathcal{N} = (V, E, s, t)$ in Figure \ref{fig:counterexample}.
Clearly, $H=\{\langle v_1, v_2 \rangle, \langle v_2, v_1\rangle \}$. However, path $p = (s, v_1, v_2, t)$
covers $\langle v_1, v_2 \rangle$ and satisfying $C_{\mathcal{N} \backslash p} =
C_{\mathcal{N}}(s, t) - 1 = 1$, and so dos path $p = (s, v_2, v_1, t)$.

The following result shows that any $k$-CF can be decomposed into
an $m$-CF and $(k-m)$-CF corresponding to different networks.

\vspace{0.3cm}
\begin{proposition} Let $\mathcal{N}=(V,E,s,t)$ be a point-to-point network.
If $F$ is a $k$-CF of $\mathcal{N}$, then $\forall \ 1\leq m \leq k-1$, $\exists F^{'}\subseteq F$ such that $F^{'}$ is a $m$-CF of $\mathcal{N}$ and $F\setminus F^{'}$ is a $(k-m)$-CF of $\mathcal{N}\setminus F^{'}$.
\end{proposition}
\begin{proof} Let  $f=\mathcal{C}_{\mathcal{N}}(s,t)$.
Denote $\mathcal{F}_{m}= \{ F^{'}\subseteq F\ |\ \mathcal{C}_{\mathcal{N}\setminus F^{'}}(s,t)= f-m \}, m=0,1,\ldots,k$. Since deleting one edge will cause the maximum flow decrease by at most $1$, we can show $\mathcal{F}_{m}\neq\emptyset$ for $m=1,2,\ldots,k-1$ by induction.



For $\forall 1\leq m\leq k-1$ we claim that $\tilde{F}\in \mathcal{F}_{m}$ with  $|\tilde{F}|=\min\{ |F'|\ \mid\ F'\in\mathcal{F}_{m}\}$ is a $m$-CF of $\mathcal{N}$. We only need to show $\forall F'\subsetneq\tilde{F},\ \mathcal{C}_{\mathcal{N}\setminus F'} \geq f-m+1$, which is true
by the minimality of $F'$.

Finally, we should prove if $F'$ is an $m$-CF of $\mathcal{N}$ and $F'\subseteq F$ then $F\setminus F'$ is a $(k-m)$-CF of $\mathcal{N}\setminus F'$. Denote $F''=F\setminus F'$. Notice $(E\setminus F')\setminus F^{''}=E\setminus (F'\cup F^{''})=E\setminus F$, we have $C_{(\mathcal{N}\setminus F')\setminus F^{''}}(s,t)=C_{\mathcal{N}\setminus F}=f-k=(f-m)-(k-m)$. Since $C_{\mathcal{N}\setminus F'}(s,t)=f-m$, we only need $\forall\tilde{F}\subsetneq F',\ C_{(\mathcal{N}\setminus F')\setminus\tilde{F}}(s,t)>f-k$. Otherwise, if $\exists \tilde{F}\subsetneq F'$ such that $C_{(\mathcal{N}\setminus F')\setminus\tilde{F}}(s,t)=f-k$, we may have $ F'\cup\tilde{F}\subsetneq F$ and $C_{\mathcal{N}\setminus(F'\cup\tilde{F})}(s,t)=f-k$ which is impossible.
\end{proof}

The following corollary is a direct generalization of the preceding result, which says a $k$-CF
can be decomposed arbitrary.

\begin{corollary} Let $\mathcal{N}=(V,E,s,t)$ be a point-to-point network, and $F$ is a $k$-CF of $\mathcal{N}$. For any $n\in \mathbb{Z}^{+}$ and $k_{i}\in \mathbb{Z}^{+}\ (i=1,2,\ldots,n)$ such that $\sum_{i=1}^{n}k_{i}=k$, then there exists pairwise disjoint sets $F_{i}\subseteq F$ such that $ \displaystyle\bigcup_{i=1}^{n}F_{i}=F$ and $F_{i}$ is a $k_{i}$-CF of $\mathcal{N}\setminus(\bigcup\limits_{j=1}^{i-1}F_{j})$.
\end{corollary}

It's natural to ask whether the converse is true, i.e., if there exists $k_{i}\in \mathbb{Z}^{+}\ (i=1,2,\ldots,n)$ and pairwise disjoint sets $F_{i}\subseteq E$ such that $F_{i}$ is a $k_{i}$-CF of $\mathcal{N}\setminus(\bigcup\limits_{j=1}^{i-1}F_{j})$, whether $\bigcup\limits_{i=1}^{n}F_{j}$ is a $(\sum\limits_{i=1}^{n}k_{i})$-CF of $\mathcal{N}$. However, the following example shows
that it is not true.

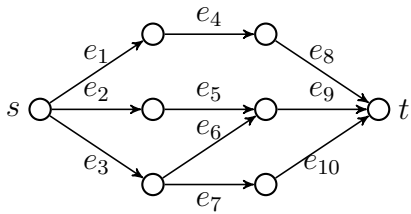
\begin{figure}[h]
\begin{center}
\begin{tikzpicture}
[inner sep=1mm,
cir/.style={circle,draw,thick},
pre/.style={->,>=stealth',semithick},
post/.style={<-,>=stealth',semithick}]
\node (s) at (-2,0) [cir,label=left:$s$]{};
\node (v1) at (-0.5,1) [cir]{};
\node (v2) at (-0.5,0) [cir]{};
\node (v3) at (-0.5,-1) [cir]{};
\node (v4) at (1,1) [cir]{};
\node (v5) at (1,0) [cir]{};
\node (v6) at (1,-1) [cir]{};
\node (t) at (2.5,0) [cir,label=right:$t$]{};
\draw [pre] (s) to node [above]{$e_1$} (v1);
\draw [pre] (s) to node [above]{$e_2$} (v2);
\draw [pre] (s) to node [below]{$e_3$} (v3);
\draw [pre] (v1) to node [above]{$e_4$} (v4);
\draw [pre] (v2) to node [above]{$e_5$} (v5);
\draw [pre] (v3) to node [above]{$e_6$} (v5);
\draw [pre] (v3) to node [below]{$e_7$} (v6);
\draw [pre] (v4) to node [above]{$e_8$} (t);
\draw [pre] (v5) to node [above]{$e_9$} (t);
\draw [pre] (v6) to node [below]{$e_{10}$} (t);
\end{tikzpicture}

\caption{The inverse of decomposition is not true}
\label{fig:inv_decc}
\end{center}
\end{figure}

See Figure \ref{fig:inv_decc}. Denote by $F_1=\{ e_4,\ e_5\}$ and $F_2=\{ e_7,\ e_9\}$. $F_1$ is a $2$-CF of $\mathcal{N}$ and $F_2$ is a $1$-CF of $\mathcal{N}\setminus F_1$. However $F_1\cup F_2$ is not a $3$-CF of $\mathcal{N}$ for $C_{\mathcal{N}\setminus \{e_4,e_7,e_9\}}(s,t)=0$.

\vspace{0.3cm}
Out next theorem asserts that all edges are contained in some $2$-CF if the maximum flow
is at least $2$.

\begin{theorem} Let $\mathcal{N}=(V,E,s,t)$ be an acyclic point-to-point network, and $C_{\mathcal{N}}(s,t) > 1$. Then all edges on some
$s$-$t$ path are contained in some $2$-CF.
\end{theorem}
\begin{proof}
Our goal is to find a path $p$ contains $e$ such that $C_{\mathcal{N}\setminus p}(s,t)\geq C_{\mathcal{N}}(s,t)-2$, then we can apply Theorem \ref{thm:main}.

If $e$ is contained in some maximum flow, i.e., there is a path $p$ containing $e$ such that $C_{\mathcal{N} \setminus p}(s,t) = C_{\mathcal{N}}(s,t)-1$, then we are done by Theorem \ref{thm:main}. If there is a path from $e$ to $t$ (or from $s$ to $e$) which meets an edge contained in a maximum flow, while any path from $s$ to $e$ (or from $e$ to t) doesn't meet any edge in flows contained under that maximum flow configuration, we can change the maximum flow to another maximum flow which contains $e$.

If there is a path from $e$ to $t$ witch meets an edge contained in a maximum flow and there is a path from $s$ to $e$ meets an edge contained in the same maximum flow, then we can find a path containing $e$ deleting which will cause the maximum flow decrease by 2. Otherwise a path from $s$ to $t$ containing $e$ won't meet an edge contained in any maximum flow, it contradicts with maximum flow.
\end{proof}

Up to now, the capacity factors we considered are restricted to point-to-point scenario. However,
 it is well known that for point to point networks, network coding provides no benefits, while major benefit of network is for multicast networks. The The following result is a simple relationship between the CF in
 a point-to-point network and multicast network.

\begin{proposition}
Let $\mathcal{N}=(V,E,s,T)$, where $T = \{t_1, t_2, \ldots, t_m\}$,
$\mathcal{N}_i=(V, E, s, t_i)$.
 $F$ is a $\overrightarrow{k}$-CF of $\mathcal{N}$, where $\overrightarrow{k}=(k_1,\ldots, k_m)$, if there exists $i \in \{1, 2, \ldots, m\}$
 such that
\begin{itemize}
\item[(1)] $F$ is a $k_i$-CF of $\mathcal{N}_i$.
\item[(2)] $C_{\mathcal{N}\backslash F}(s, t_j) = C_{\mathcal{N}}(s, t_j)-k_j$ for all $j \in \{1,2,\ldots, m\}.$
\end{itemize}
\end{proposition}
\begin{proof} If $F$ satisfies condition (1) and (2), then $F$ is a CF because for any proper subset $F'$ of $F$,
$C_{\mathcal{N}\backslash F'}(s, t_i) > C_{\mathcal{N}}(s, t_i)-k_i$ for $F$ is a $k_i$-CF of $\mathcal{N}_i$.
\end{proof}

A moment thought reveals that the converse is not true.

\section{Algorithm to Compute $D$-set and $H$-set}

From Corollary \ref{cor:main}, an edge is in $D$-set if and only if it
is contained in some maximum flow configuration. The following
algorithm gives a method to solve the problem.

\begin{algorithm}
\label{alg:main} The input is a point-to-point network
$\mathcal{N}=(V,E,s,t)$. The output is the $D$-set and
$H$-set of the network.
\begin{enumerate}
\item \emph{[Initialization]} $D=\O$, $H=\O$.
\item \emph{[Maximum flow]} Find a maximum flow $f$ on
$\mathcal{N}$ and obtain the corresponding residual network
$\mathcal{N}_f$. \item \emph{[Choose an edge]} If there is an edge
$\langle u, v \rangle \in E$ and $\langle u, v \rangle \not \in
D \cup H$, then choose $\langle u, v \rangle$
and go to step 4, else go to step 6. \item \emph{[$\langle u, v
\rangle$ is in $f$?]} If $\langle u, v \rangle$ is in $f$, then
$D$ $\leftarrow$$D$ $\cup$ $\{\langle u, v
\rangle\}$ and go to step 3, else go to step 5. \item \emph{[A
circle containing $\langle u, v \rangle$?]} Since $\langle u, v
\rangle$ is not in $f$, $\langle u, v \rangle$ is a forward edge in
the residual network $\mathcal{N}_f$. If there is a path from $v$ to
$u$ in $\mathcal{N}_f$, then $ D \leftarrow D
\cup \{ \langle u, v \rangle \}$, else $H \leftarrow
H \cup \{ \langle u, v \rangle \}$. Go to step 3. \item
\emph{[End]} $D$ and $H$ are the
$D$-set and $H$-set of network $\mathcal{N}$
respectively.
\end{enumerate}
\end{algorithm}

\begin{figure}[h]
\begin{center}
\begin{tikzpicture}
[inner sep=1mm,
cir/.style={circle,draw,thick},
pre/.style={->,>=stealth',semithick},
post/.style={dotted,->,>=stealth',semithick}]
\node (v1) at (0,0) [cir,label=above:$v_1$]{};
\node (v2) at (-1.5,-1) [cir,label=left:$v_2$]{};
\node (v3) at (-0.5,-1) [cir,label=left:$v_3$]{};
\node (v4) at (0.5,-1) [cir,label=right:$v_4$]{};
\node (v5) at (1.5,-1) [cir,label=right:$v_5$]{};
\node (v6) at (-1.5,-2) [cir,label=left:$v_6$]{};
\node (v7) at (-0.5,-2) [cir,label=left:$v_7$]{};
\node (v8) at (0.5,-2) [cir,label=right:$v_8$]{};
\node (v9) at (1.5,-2) [cir,label=right:$v_9$]{};
\node (v10) at (0,-3) [cir,label=below:$v_{10}$]{};
\draw [pre] (v1) to [bend right=10] (v2);
\draw [pre] (v1) to (v3);
\draw [pre] (v1) to (v4);
\draw [pre] (v1) to [bend left=10](v5);
\draw [pre] (v2) to (v6);
\draw [pre] (v3) to (v7);
\draw [pre] (v4) to (v8);
\draw [pre] (v5) to (v9);
\draw [pre] (v6) to [bend right=10](v10);
\draw [pre] (v7) to (v10);
\draw [pre] (v8) to (v10);
\draw [pre] (v9) to [bend left=10](v10);
\draw [pre] (v7) to (v2);
\draw [pre] (v8) to (v3);
\draw [pre] (v9) to (v4);

\node (v1) at (0+4.5,0) [cir,label=above:$v_1$]{};
\node (v2) at (-1.5+4.5,-1) [cir,label=left:$v_2$]{};
\node (v3) at (-0.5+4.5,-1) [cir,label=left:$v_3$]{};
\node (v4) at (0.5+4.5,-1) [cir,label=right:$v_4$]{};
\node (v5) at (1.5+4.5,-1) [cir,label=right:$v_5$]{};
\node (v6) at (-1.5+4.5,-2) [cir,label=left:$v_6$]{};
\node (v7) at (-0.5+4.5,-2) [cir,label=left:$v_7$]{};
\node (v8) at (0.5+4.5,-2) [cir,label=right:$v_8$]{};
\node (v9) at (1.5+4.5,-2) [cir,label=right:$v_9$]{};
\node (v10) at (0+4.5,-3) [cir,label=below:$v_{10}$]{};
\draw [pre,very thick] (v1) to [bend right=10] (v2);
\draw [pre,very thick] (v1) to (v3);
\draw [pre,very thick] (v1) to (v4);
\draw [pre] (v1) to [bend left=10](v5);
\draw [pre,very thick] (v2) to (v6);
\draw [pre,very thick] (v3) to (v7);
\draw [pre,very thick] (v4) to (v8);
\draw [pre] (v5) to (v9);
\draw [pre,very thick] (v6) to [bend right=10](v10);
\draw [pre,very thick] (v7) to (v10);
\draw [pre,very thick] (v8) to (v10);
\draw [pre] (v9) to [bend left=10](v10);
\draw [pre] (v7) to (v2);
\draw [pre] (v8) to (v3);
\draw [pre] (v9) to (v4);

\node (v1) at (0+9,0) [cir,label=above:$v_1$]{};
\node (v2) at (-1.5+9,-1) [cir,label=left:$v_2$]{};
\node (v3) at (-0.5+9,-1) [cir,label=left:$v_3$]{};
\node (v4) at (0.5+9,-1) [cir,label=right:$v_4$]{};
\node (v5) at (1.5+9,-1) [cir,label=right:$v_5$]{};
\node (v6) at (-1.5+9,-2) [cir,label=left:$v_6$]{};
\node (v7) at (-0.5+9,-2) [cir,label=left:$v_7$]{};
\node (v8) at (0.5+9,-2) [cir,label=right:$v_8$]{};
\node (v9) at (1.5+9,-2) [cir,label=right:$v_9$]{};
\node (v10) at (0+9,-3) [cir,label=below:$v_{10}$]{};
\draw [post] (v1) to [bend right=10] (v2);
\draw [post] (v1) to (v3);
\draw [post] (v1) to (v4);
\draw [pre] (v1) to [bend left=10](v5);
\draw [post] (v2) to (v6);
\draw [post] (v3) to (v7);
\draw [post] (v4) to (v8);
\draw [pre] (v5) to (v9);
\draw [post] (v6) to [bend right=10](v10);
\draw [post] (v7) to (v10);
\draw [post] (v8) to (v10);
\draw [pre] (v9) to [bend left=10](v10);
\draw [pre] (v7) to (v2);
\draw [pre] (v8) to (v3);
\draw [pre] (v9) to (v4);
\end{tikzpicture}
\caption{Network, maximum flow, and residual network}
\label{fig:3flow}
\end{center}
\end{figure}
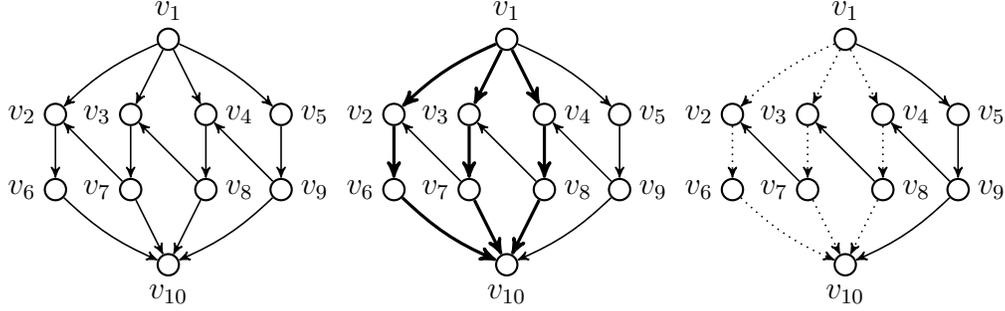

Consider the network $\mathcal{N} = (V,E,s,t)$ shown in
Figure \ref{fig:3flow}, the maximum flow value is $3$. One
maximum flow scheme $f$ and the corresponding residual network
$\mathcal{N}_f$ are also shown in Figure \ref{fig:3flow}.

Since edges (in bold) in the maximum flow $f$ is in the
$D$-set, $\langle v_1, v_2 \rangle$, $\langle v_1, v_3
\rangle$, $\langle v_1, v_4 \rangle$, $\langle v_2, v_6 \rangle$,
$\langle v_3, v_7 \rangle$, $\langle v_4, v_8 \rangle$, $\langle
v_6, v_{10} \rangle$, $\langle v_7, v_{10} \rangle$, $\langle v_8,
v_{10} \rangle$ $\in D$. Let's consider the remaining
edges. For $\langle v_1, v_5 \rangle$, in the residual network
$\mathcal{N}_f$, there is no cycle containing it. So, $\langle v_1,
v_5 \rangle \in H$. Similar, there is no cycle in
$\mathcal{N}_f$ containing $\langle v_9, v_5 \rangle$, $\langle v_8,
v_3 \rangle$, $\langle v_7, v_2 \rangle$, which implies $\langle
v_9, v_5 \rangle$, $\langle v_8, v_3 \rangle$, $\langle v_7, v_2
\rangle$ $\in H$. For $\langle v_4, v_9 \rangle$ and
$\langle v_9, v_{10} \rangle$, there is a cycle $(v_4, v_9, v_{10},
v_8)$ containing them. So $\langle v_4, v_9 \rangle$, $\langle v_9,
v_10 \rangle$ $\in D$.

To sum up, the $D$-set and $H$-set of network
$\mathcal{N}(V,E,s,t)$ is
\begin{displaymath}
\begin{split}
D =& \{\langle v_1, v_2 \rangle, \langle v_1, v_3 \rangle,
\langle v_1, v_4 \rangle, \langle v_2, v_6 \rangle, \langle v_3, v_7
\rangle, \langle v_4, v_8 \rangle,\\& \quad\quad\quad \langle v_6,
v_{10} \rangle, \langle v_7, v_{10} \rangle, \langle v_8, v_{10}
\rangle, \langle v_4, v_9 \rangle, \langle v_9, v_{10} \rangle\}
\end{split}
\end{displaymath}
and
\begin{displaymath}
H = \{\langle v_1, v_5 \rangle, \langle v_9, v_5 \rangle,
\langle v_8, v_3 \rangle, \langle v_7, v_2 \rangle \}.
\end{displaymath}

The next theorem shows the correctness of algorithm \ref{alg:main},
which reduces the existence of a maximum flow containing $e$ to
the existence of a cycle containing $e$ in the residual network
corresponding to an arbitrary maximum flow.

\begin{theorem}
Let $\mathcal{N}=(V,E,s,t)$ be a point-to-point network and $f$ be a
maximum flow on $\mathcal{N}$. The corresponding residual network is
$\mathcal{N}_f$. For an edge $e \in E$, $e$ is in a maximum flow on
$\mathcal{N}$ for some $f'$, if and only if $e$ is in $f$ or there
is a cycle in $\mathcal{N}_f$ containing $e$.
\end{theorem}
\begin{proof} \textbf{Necessity:} When $e$ is in $f$, it's obvious.
Assume $e$ is in a cycle $C$ and $C$ is in $\mathcal{N}_f$. Adding a
cyclic flow $C$ in $\mathcal{N}_f$ can generate another flow $f$,
which is also a maximum flow. Because $e$ is a forward edge in
$\mathcal{N}_f$, $e$ becomes a reversal edge in $\mathcal{N}_{f'}$,
which means there is a maximum flow containing $e$.

\textbf{Sufficiency:} Suppose $e$ is in a maximum flow $f'$. If $f =
f'$, then $e$ is in $f$. If not, subtract flow $f$ from $f'$,
denoted by $f-f'$. Since subtracting does not break the conservation
constraints, $f-f'$ is a feasible flow. Since $|f| = |f'|$, the flow
value of $f-f'$ is $0$. Therefore, $f-f'$ can be decomposed into one
or more cycles. Because $e$ is in $f'$ and not in $f$, $e$ is one
cycle of the flow $f-f'$, denoted by $C$. It it easy to see cycle
$C$ is in the residual network $\mathcal{N}_f$.
\end{proof}

How to find a cycle in $\mathcal{N}_f$ containing an
edge $\langle u, v \rangle$? Note that all edges of network
$\mathcal{N}$ considering in this paper have unit capacity, and
$\langle u, v \rangle$ is a forward edge in $\mathcal{N}_f$. It's
easy to see that there is a cycle containing $\langle u, v \rangle$
if and only if there is path from $v$ to $u$. Therefore, the problem
searching cycles in $\mathcal{N}_f$ is reduced to the connectivity
problem of two vertices in a digraph.

Now we analyze the time complexity of this algorithm. This algorithm
could be divided into two separated parts, maximum flow and
searching cycles. If we find a maximum flow by using
relabel-to-front algorithm \cite{GT86}, whose running time is
$O({|V|}^3)$, and using Floyd-Warshall algorithm \cite{CLRS01} to
implement searching cycles for all edges, whose running time is also
$O({|V|}^3)$, then the total running time is
$O({|V|}^3)+O({|V|}^3)=O({|V|}^3)$. If we have a maximum flow $f$
preserved and just wondering the belonging of one edge, the time
complexity is $O(|E|)$, where determining the connectivity of two
vertices just needs a depth-first-search or breadth-first-search
through all vertices and edges.

\section{Computational Hardness Relating to Capacity Factors}

We consider the Maximum Capacity Factor (MCF) problem for a
point-to-point network which might have cycles. Given a
point-to-point network $\mathcal{N} = (V,E,s,t)$ and a specific
number $k$, our goal is to answer whether there is a capacity factor
with size not less than $k$. The formal language for this decision
problem is: MCF$= \{ \langle \mathcal{N}, k \rangle :$ $\mathcal{N} = \langle V, E, s, t
\rangle$ is a network with some capacity factor with size greater
than or equal to $k \}$.

In the following proof of our theorem, we will reduce a known
NP-complete problem to MCF, which is NAESAT (stands for
``not-all-equal''). NAESAT is an variant of SAT. In NAESAT, we are
given a set of clauses with three literals, and we insist that in no
clause are all three literals equal in truth value, i.e., neither
all true, nor all false. It is known that NAESAT is NP-complete
\cite{Sis06}.

Before proving MCF is NP-complete, we present a definition and a
proposition which will be used in the following proof.

\begin{definition}\cite{KW07}
Let $\mathcal{N} = (V,E,s,t)$ be a point-to-point network. An
$s$-$t$ cut $[V_1, \overline{V_1}]$ of $\mathcal{N}$ is a partially
connected $s$-$t$ cut if for any $e=\langle u, v \rangle \in [V_1,
\overline{V_1}]$, there is a path from $s$ to $u$ in
$\mathcal{N}(V_1)$ and there is a path form $v$ to $t$ in
$\mathcal{N}(\overline{V_1})$, where $\mathcal{N}(V_1)$ and
$\mathcal{N}(\overline{V_1})$ are the vertex-induced network of
$\mathcal{N}$ by vertex sets $V_1$ and $\overline{V_1}$
respectively.
\end{definition}

It's worth noting that a partially connected $s$-$t$ cut $[V_1,
\overline{V_1}]$ does not necessarily mean for every vertex $u \in
V_1$, there is a path from $s$ to $u$ in $\mathcal{N}(V_1)$, and for
every vertex $v \in \overline{V_1}$, there is a path from $v$ to $t$
in $\mathcal{N}(\overline{V_1})$. In other words, $\mathcal{N}(V_1)$
and $\mathcal{N}(\overline{V_1})$ might not be connected graphs.

In \cite{KW07}, it's proved that the size of a capacity factor of
network $\mathcal{N} = (V,E,s,t)$ is upper-bounded by the size of
the maximal partially-connected $s$-$t$ cut minus
$C_{\mathcal{N}}(s,t) - 1$. In unit capacity network, there is a
one-on-one correspondence between capacity factors and partially-connected
cut, as the following proposition reveals.

\begin{proposition}
\label{pros:par_cut} Let $\mathcal{N} = (V,E,s,t)$ be a
point-to-point network with unit capacity, i.e.,
$C_{\mathcal{N}}(s,t) = 1$. $F$ is a capacity factor if and only if
$F$ is a partially connected $s$-$t$ cut of $\mathcal{N}$.
\end{proposition}
\begin{proof} \textbf{Necessity:} Assume $F = [V_1,
\overline{V_1}]$, where $[V_1, \overline{V_1}]$ is a
partially-connected $s$-$t$ set. Since $F$ is a cut, we
have $C_{\mathcal{N} \backslash F}(s,t) = 0$. By the definition of a
partially-connected $s$-$t$ cut, for any edge $e = \langle u, v
\rangle \in F$, there exist a path from $s$ to $u$ in
$\mathcal{N}(V_1)$ and a path from $v$ to $t$ in
$\mathcal{N}(\overline{V_1})$. Therefore, $C_{\mathcal{N} \backslash
F \cup \{ e \}}(s,t) = 1$, which implies that $F$ is a capacity
factor.

\textbf{Sufficiency:} Assume $F$ is a capacity factor of
$\mathcal{N}$, we have $C_{\mathcal{N} \backslash F}(s,t) \leq
C_{\mathcal{N}}(s,t) - 1 = 0$. Therefore, $s$ and $t$ are
disconnected in the network $\mathcal{N} \backslash F$. Denote the
vertices reachable from $s$ in $\mathcal{N} \backslash F$ by $V_1$
(including $s$), the vertices that could reach $t$ in $\mathcal{N}
\backslash F$ by $V_2$ (including $t$), and the remaining ones by
$V_3$.

Since $F$ is a capacity factor, adding an arbitrary edge of $F$ in
the network $\mathcal{N} \backslash F$ will make $s$ and $t$
connected, it's clear that all edges in $F$ should be of the form
$\langle u, v \rangle$, where $u \in V_1$ and $v \in V_2$, which
implies $F \subset [V_1, V_2]$. Having considered $V_1$ and $V_2$
are disconnected, we conclude $F = [V_1, V_2]$. Since both $V_1$,
$V_3$ and $V_3$, $V_2$ are disconnected, we have $[V_1, V_3] = [V_3,
V_2] = \emptyset$. Thus, $F = [V_1, V_2] = [V_1, V_2 \cup V_3] =
[V_1, \overline{V_1}]$, which implies $F$ is a partially-connected
$s$-$t$ cut.
\end{proof}

As we have the above characterization of capacity factor in unit-flow
network, the following reduction is similar with the reduction from
NAESET to maximal cut \cite{Sis06}.

\vspace{0.3cm}
\begin{theorem}
\label{thm:NPC} MCF is NP-complete.
\end{theorem}
\textbf{Proof:} Firstly, we claim MCF is in NP. Providing the verifier of a capacity $F$ with $|F| \geq m$,
it's easy to check $F$ is a capacity factor by testing
$C_{N\backslash F}(s, t) = C_{N \backslash F}(s, t) - 1$ and
$C_{N\backslash F \cup e}(s,t) = C_{N\backslash F}(s,t)$ for every $e
\in F$, which can be done by running network flow algorithm for $|F|+1$
times. Therefore, given such a proof, there is a verifier in
polynomial time, which implies that MCF is in NP.

Secondly, we shall reduce NAESAT to MCF. Given an expression $T$
consist of $m$ clauses with three literals each, we will construct a
network $\mathcal{N} = (V, E, s, t)$ and an integer $k$, such that
$T$ is in NAESAT if and only if $\langle \mathcal{N}, k \rangle$ is
in MCF.

Suppose that the clauses are $C_1, \ldots, C_m$, and the variables
appearing in them $x_1, \ldots, x_n$. At first, we set $2n$
vertices, which are denoted by $x_1$, $\ldots$, $x_n$, $\neg x_1$,
$\ldots$, $\neg x_n$, and four additional vertices $s, s', t,
t'$, where $s, t$ are the source and sink separately.

Then we add some edges which can be classified into the following 3
categories.
\begin{itemize}
\item \textbf{Crossing edges:} For each clause $C_i = (x \vee y \vee z)$ add the following
bidirectional edges $\langle x, y \rangle$, $\langle x, z \rangle$,
$\langle y, z \rangle$ respectively if no self-loop are created.
Note that $x = y$ might happen, whereas $x = y = z$ is impossible,
which implies there are at least two and at most three bidirectional
edges will be created for each clause.
\item \textbf{Forcing edges:} For each pair of $\langle x_i, \neg x_i
\rangle$, $i=1,2,\ldots, n.$, add $4m$ bidirectional edges between
them. Add $6mn$ edges between $s'$ and $t'$.
\item \textbf{Connecting edges:} Add edges $\langle s', x_i \rangle$, $\langle s', \neg x_i
\rangle$, $\langle x_i, t \rangle$, $\langle \neg x_i, t \rangle$,
$i = 1, 2, \ldots, n$. Add edges $\langle s, s' \rangle$,
$\langle t, t' \rangle$.
\end{itemize}
Finally, set $k = 10mn+2m+2n$ and our construction is complete.

Now, we show that expression $T$ is in NAESAT if and only if the
corresponding constructed $\langle \mathcal{N}, k \rangle$ is in
MCF.

\textbf{T $\in$ NAESAT $\Rightarrow$ $\langle \mathcal{N}, k \rangle \in$ MCF: }Put the vertices of true literals on the left hand side
with $s$ and $s'$, while put those of false on the right hand
side with $t'$ and $t$. Having considered the fact that each
clause has both true and false literal(s), there will be exactly two
``crossing edges'' across two piles of vertices. Taking account of
all the ``forcing edges'' and ``connecting edges'', there will be
exactly $6mn+4mn+2m+2n = k$ edges in the cut. In addition, it is
clearly checked that adding an arbitrary edge in the cut will make
$s$ and $t$ connected.

\textbf{$\langle \mathcal{N}, k \rangle \in$ MCF $\Rightarrow$ T $\in$ NAESAT: } By Proposition \ref{pros:par_cut},
we know that any capacity factor $F$ of $\mathcal{N}$ is a partially
connected $s$-$t$ cut. Hence, we consider the maximum possible
partially connected $s$-$t$ cut of $\mathcal{N}$ instead. Suppose
$[V_1, \overline{V_1}]$ is a partially connected $s$-$t$ cut with
maximum size. First of all, we claim $s' \in V_1$ and
$t' \in \overline{V_1}$, since $6mn$ number of edges have
overwhelming impact on the size of $[V_1, \overline{V_1}]$.
Secondly, we claim $x_i$ and $\neg x_i$ must lie in different sides
of the cut, since there are $4m$ bidirectional edges between them,
which also has overwhelming impact on the size of cut when
$s'$ and $t'$ are fixed. After that, the contribution
of edges $\langle s', x_i \rangle$, $\langle s'$,
$\neg x_i \rangle$, $\langle x_i, t' \rangle$, $\langle \neg
x_i, t' \rangle$, $i = 1, 2, \ldots, n$, to the size of
$[V_1, \overline{V_1}]$ is fixed. No matter $x_i \in V_1$ or $\neg
\in V_1$, there are exactly two edges in $[V_1, \overline{V_1}]$ for
each $i$. Finally, we consider crossing edges for  each clause. If
$C_i=(x\vee y\vee z)$ with $x\neq y$ and $y \neq z$, among six edges
$\langle x, y \rangle$, $\langle x, z \rangle$, $\langle y,
x\rangle$, $\langle y, z \rangle$, $\langle z, x \rangle$, $\langle
z, y \rangle$, at most two of them could be in $[V_1,
\overline{V_1}]$, if $x, y, z$ do not lie on the same side of the
cut. If $C_i = (x\vee x\vee y)$ with $x \neq y$, among $\langle x, y
\rangle_2$, $\langle y, x \rangle_2$, there are also at most two of
them could be in $[V_1, \overline{V_1}]$ if $x$ and $y$ lie in
different side of the cut. From above discussion, we could see $|F|$
is upper-bounded by $6mn+4mn+2n+2m = k$. The the number is achieved
if for each clause $C_i = (x\vee y\vee z)$, $x, y, z$ do not lie in
the same side of the cut. Assigning the literals in $V_1$ by true and those in $\overline{V_1}$ by
false, we will see it's an assignment that renders $T \in$ NAESAT.
$\hfill{} \blacksquare$

\vspace{0.3cm} To demonstrate the construction of Theorem
\ref{thm:NPC}, we present the following example.

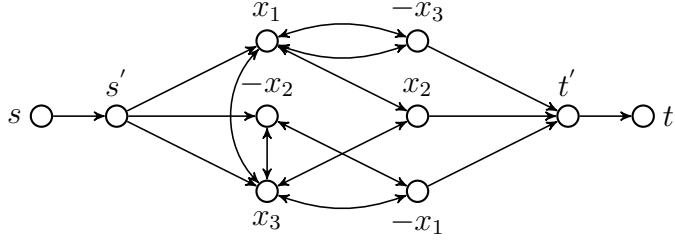
\begin{figure}
\begin{center}
\begin{tikzpicture}
[inner sep=1mm,
cir/.style={circle,draw,thick},
pre/.style={->,>=stealth',semithick},
bi/.style={<->,>=stealth',semithick}]
\node (s) at (0,0) [cir,label=left:$s$]{};
\node (v0) at (1,0) [cir,label=above:$s^{'}$]{};
\node (v1) at (3,1) [cir,label=above:$x_{1}$]{};
\node (v2) at (5,1) [cir,label=above:$-x_{3}$]{};
\node (v3) at (3,0) [cir,label=above:$-x_{2}$]{};
\node (v4) at (5,0) [cir,label=above:$x_{2}$]{};
\node (v5) at (3,-1) [cir,label=below:$x_{3}$]{};
\node (v6) at (5,-1) [cir,label=below:$-x_{1}$]{};
\node (v7) at (7,0) [cir,label=above:$t^{'}$]{};
\node (t) at (8,0) [cir,label=right:$t$]{};
\draw [pre] (s) to (v0);
\draw [pre] (v0) to (v1);
\draw [pre] (v0) to (v3);
\draw [pre] (v0) to (v5);
\draw [bi] (v1) to [bend right=45] (v5);
\draw [bi] (v1) to [bend right=20] (v2);
\draw [bi] (v1) to [bend left=20] (v2);
\draw [bi] (v1) to (v4);
\draw [bi] (v3) to (v6);
\draw [bi] (v5) to (v4);
\draw [bi] (v5) to [bend right=20] (v6);
\draw [bi] (v3) to (v5);
\draw [pre] (v2) to (v7);
\draw [pre] (v4) to (v7);
\draw [pre] (v6) to (v7);
\draw [pre] (v7) to (t);
\end{tikzpicture}
\caption{Reduction from NAESAT to MCF}
\label{fig:NPC}
\end{center}
\end{figure}

Given the expression $T = (x_1 \vee x_2 \vee x_3)\wedge(x_1 \vee
x_3 \vee \neg x_3)\wedge(\neg x_1 \vee \neg x_2 \vee x_3)$, we
construct a network $\mathcal{N} = (V, E, s, t)$ according to the
proof of Theorem \ref{thm:NPC} shown in Figure \ref{fig:NPC}. Note
that all the ``forcing edges'', i.e. the edges between $x_i$ and
$\neg x_i$, the edges between $s'$ and $t'$, as well
as the ``connecting edges'' $\langle s', \neg x_1 \rangle$,
$\langle s', x_2 \rangle$, $\langle s', \neg x_3
\rangle$, $\langle x_1, t' \rangle$, $\langle \neg x_2,
t' \rangle$, $\langle x_3, t' \rangle$ are not drawn
in the figure. In this case, $m = 3$, $n = 3$ and $k = 96$.

It's clear that $x_1 = 1$, $x_2 = 0$ and $x_3 = 1$ is an assignment
makes $T$ in NAESAT. If we put vertices $x_1$, $\neg x_2$, $x_3$ on
the left and $\neg x_1$, $x_2$, $\neg x_3$ on the right, just as
what is drawn on Figure \ref{fig:NPC}, we obtain a maximum partially
connected $s$-$t$ cut $[V_1, \overline{V_1}]$ with size $k$, where
$V_1 = \{ s, s', x_1, \neg x_2, x_3 \}$.

\vspace{0.5cm}
Since MCF is NP-complete problem, deciding the maximum capacity
factor containing some specific edge is also NP-complete. Otherwise,
by enumerating all the edges of in the latter problem, we can solve
MCF in polynomial time, which is a contradiction.

Compared to the problem of the maximum-sized capacity factor,
calculating the capacity rank seems to be more important. As far as
we know, there is no polynomial time algorithm to calculate the
capacity rank in a general network. Furthermore, we don't know
whether it is in NP-complete. However, there are some evidences
indicating this problem is not easy.

\begin{theorem} If the
capacity rank of an arbitrary edge in a unit capacity network can be computed in time $f(|V|,
|E|)$, then the maximum flow can be solved in time $f(|V|+2, 2|E|+2)$ for
any any point-to-point network $\mathcal{N} = (V, E, s, t)$.

In other words, the time complexity of calculating capacity rank is
lower bounded by calculating the maximum flow.
\end{theorem}
\begin{proof} Assume there is an algorithm to compute the capacity
rank of an arbitrary edge in any unit-capacity point-to-point network in time
$f(|V|, |E|)$. For a point-to-point network $\mathcal{N} = (V, E, s,
t)$ of maximum flow problem, we construct a corresponding network
$\mathcal{N}=(V',E',s',t')$ and prove that $CR_{\mathcal{N}'}(e) =
C_{\mathcal{N}}(s,t)$, where $e := \langle s', t \rangle$.

The construction is as follows: Let
$$V' = V \cup \{s', t'\}$$
and
$$E' = E \cup \{\langle s', t \rangle\, \langle t, t' \rangle, \langle s', s
\rangle_{|E|}\},$$
where $\langle s', s \rangle_{|E|}$ denotes $|E|$ different edges from $s'$ to $s$. It's easy
to show $CR_{\mathcal{N'}}(\langle s', r \rangle) =
C_{\mathcal{N}}(s,t)$, which is left to the reader.

\end{proof}

Many questions in graph theory about edges have
natural analogues for vertices\cite{Wes01}, and the vertices version
 is often harder than that of edges. For example, Eulerian
circuit is defined as a closed trail containing all edges, whereas
Hamilton cycle is a closed path visiting all the vertices exactly
once; independent set has no adjacent vertice, whereas matching has
no ``adjacent'' edges. However, deciding whether a graph has a
Eulerian circuit is easy, while deciding whether a graph has a
Hamilton path is NP-complete; finding a maximum independent set is
NP-hard, while finding a maximum matching has a polynomial time
algorithm. It's natural to define the vertex capacity factor and
investigate the relationship between (edge) capacity factor.

\begin{definition}
Let $\mathcal{N}=(V,E,s,t)$ be a point-to-point network. A nonempty
subset $F$ of $V$ is a vertex capacity factor of $\mathcal{N}$ if and only if
the following two conditions hold:
\begin{enumerate}
\item $C_{\mathcal{N} \backslash V}(s,t) < C_{\mathcal{N}}(s,t)$;
\item $C_{\mathcal{N} \backslash V'}=C_{\mathcal{N}}(s,t)$ for any
proper subset $V' of F$.
\end{enumerate}
$\mathcal{N} \backslash V$ denotes the induced network formed by
deleting $V$ and all the edges adjacent to edges in $V$ in
$\mathcal{N}$. Similarly, the vertex capacity rank of a vertex $v$
is defined as the minimum size of the vertex capacity factor
containing $v$.
\end{definition}

Just as many analogue problem on vertices and edges, the vertex version
captures the edge version through line graph.

 For $\mathcal{N} = (V, E, s, t)$, the line
network $\mathcal{N}' = (V', E', s', t')$ as follows.
\begin{itemize}
\item $V' = \{e^{\text{in}}, e^{\text{out}} : e \in E\} \cup \{s', t'\}$.
\item $E' = \{\langle e^{\text{in}}, e^{\text{out}} \rangle : e \in E\} \cup
\{ \langle e^{\text{out}}_1, e^{\text{in}}_2 \rangle : \textrm{head}(e_1) =
\textrm{tail}(e_2), e_1, e_2 \in E \} \cup \{ \langle
s', e^{\text{in}} \rangle : \textrm{tail}(e) = s, e \in E\} \cup \{
(e^{\text{out}}, t') : \textrm{head}(e) = t, e \in E \}$.
\end{itemize}
Slightly different from the definition of line graph, we split a vertex
representing an edge $e$ of $E$ in two, say $e^{\text{in}}$ and $e^{\text{out}}$,
and add an directional edge $\langle e^{\text{in}}, e^{\text{out}} \rangle$. This
modification guarantees that the capacity of each vertex is upper
bounded by $1$.

Now we will show $F = \{e_1, e_2, \ldots, e_m \}$ is a capacity
factor in $\mathcal{N}$ if and only if $F' = \{e'_1, e'_2, \ldots,
e'_m \mid e_i \in \{e^{\text{in}}_i, e^{\text{out}}_i \}, i=1,2,\ldots,m \}$ is a
vertex capacity in $\mathcal{N'}$. The key fact contributing to the
above conclusion is that any $m$ edge-disjoint $s$-$t$ paths in
$\mathcal{N}$ corresponds to $m$ vertex-disjoint $s'$-$t'$ paths in
$\mathcal{N'}$ (except the starting vertex $s'$ and ending vertex
$t'$), and the correspondence is one-to-one.

For one direction, assuming $F = \{e_1, e_2, \ldots, e_m \}$ is a
capacity factor in $\mathcal{N}$, let $F' = \{e'_1, e'_2, \ldots,
e'_m : e_i \in \{e^{\text{in}}_i, e^{\text{out}}_i \}, i=1,2,\ldots,m \}$.
Since $C_{\mathcal{N} \backslash F}(s,t) < C_{\mathcal{N}}(s,t)$,
for any maximal flow $f$ on $\mathcal{N}$, the intersection of $F$
and $f$ is not empty. Since the capacity of edges in $\mathcal{N}$
are all integers, $f$ can be decomposed into $|f|$ edge-disjoint
$s$-$t$ paths, which corresponds to $|f|$ vertex-disjoint $s'$-$t'$
paths in $\mathcal{N'}$ and vice versa. Therefore, the intersection
of $F'$ and arbitrary set of $|f|$ vertex-disjoint paths in
$\mathcal{N}'$ is not empty, which implies $C_{\mathcal{N'}
\backslash F'}(s,t) < C_{\mathcal{N'}}(s,t)$. For any subset $G$ of
$F$, since $C_{\mathcal{N} \backslash G}(s,t) =
C_{\mathcal{N}}(s,t)$, there exists a set of $C_{\mathcal{N}}(s,t)$
edge-disjoint paths, which has no common edges of $G$. Therefore,
$G'$ has  no common vertices with the corresponding
$C_{\mathcal{N}}(s,t)=C_{\mathcal{N}'}(s,t)$ vertex-disjoint paths
in $\mathcal{N}'$, which implies $C_{\mathcal{N}' \backslash
G'}(s,t) = C_{\mathcal{N'}}(s,t)$. Thus, $F'$ is a vertex capacity
factor of $\mathcal{N}'$. Using the one-on-one correspondence, the
other direction is similar to prove.

Following is a concrete example. Consider the network $\mathcal{N}=(V,E,s,t)$ on the top of
Figure \ref{fig:dual}. The corresponding dual network
$\mathcal{N'}=(V',E',s',t')$ is shown on the bottom. There are 7
capacity factors in $\mathcal{N}$, which are $\{e_1\}$, $\{e_2\}$,
$\{e_7\}$, $\{e_3, e_5\}$, $\{e_3, e_6\}$, $\{e_4, e_5\}$, $\{e_4,
e_6\}$. And there are 7 classes of vertex capacity factors in
$\mathcal{N}'$, which are exactly $\{e'_1\}$, $\{e'_2\}$,
$\{e'_7\}$, $\{e'_3, e'_5\}$, $\{e'_3, e'_6\}$, $\{e'_4, e'_5\}$,
$\{e'_4, e'_6\}$, where $e'_i \in \{ e_i^{\text{in}}, e_i^{\text{out}} \}$. And
there is a one-to-one correspondence between them.
\begin{figure}[h!]
\begin{center}
\begin{tikzpicture}
[inner sep=1mm,
cir/.style={circle,draw,thick},
pre/.style={->,>=stealth',semithick}]
\node (s) at (0,0) [cir,label=left:$s$]{};
\node (v1) at (2,1) [cir]{};
\node (v2) at (2,0) [cir]{};
\node (v3) at (4,0) [cir]{};
\node (v4) at (4,-1) [cir]{};
\node (t) at (6,0) [cir,label=right:$t$]{};
\draw [pre] (s) to [bend left=15] node [above]{$e_1$} (v1);
\draw [pre] (s) to node [below]{$e_2$} (v2);
\draw [pre] (v2) to node [above]{$e_3$} (v3);
\draw [pre] (v3) to node [above]{$e_4$} (t);
\draw [pre] (v1) to [bend left=15] node [above]{$e_7$} (t);
\draw [pre] (v2) to [bend right=15] node [below]{$e_5$} (v4);
\draw [pre] (v4) to [bend right=15] node [below]{$e_6$} (t);

\node (s) at (-1,-4) [cir,label=left:$s^{'}$]{};
\node (v21) at (0,-4) [cir,label=above:$e_2^i$]{};
\node (v22) at (1,-4) [cir,label=above:$e_2^o$]{};
\node (v31) at (2,-4) [cir,label=above:$e_3^i$]{};
\node (v32) at (3,-4) [cir,label=above:$e_3^o$]{};
\node (v41) at (4,-4) [cir,label=above:$e_4^i$]{};
\node (v42) at (5,-4) [cir,label=above:$e_4^o$]{};
\node (v11) at (0,-2.5) [cir,label=above:$e_1^i$]{};
\node (v12) at (1,-2.5) [cir,label=above:$e_1^o$]{};
\node (v71) at (3,-2.5) [cir,label=above:$e_7^i$]{};
\node (v72) at (4,-2.5) [cir,label=above:$e_7^o$]{};
\node (v51) at (2,-5.5) [cir,label=above:$e_5^i$]{};
\node (v52) at (3,-5.5) [cir,label=above:$e_5^o$]{};
\node (v61) at (4,-5.5) [cir,label=above:$e_6^i$]{};
\node (v62) at (5,-5.5) [cir,label=above:$e_6^o$]{};
\node (t) at (6,-4) [cir,label=right:$t^{'}$]{};
\draw [pre] (s) to (v21);
\draw [pre] (v21) to (v22);
\draw [pre] (v22) to (v31);
\draw [pre] (v31) to (v32);
\draw [pre] (v32) to (v41);
\draw [pre] (v41) to (v42);
\draw [pre] (v42) to (t);

\draw [pre] (s) to [bend left=30] (v11);
\draw [pre] (v11) to (v12);
\draw [pre] (v12) to (v71);
\draw [pre] (v71) to (v72);
\draw [pre] (v72) to [bend left=30] (t);

\draw [pre] (v22) to [bend right=30] (v51);
\draw [pre] (v51) to (v52);
\draw [pre] (v52) to (v61);
\draw [pre] (v61) to (v62);
\draw [pre] (v62) to [bend right=30] (t);

\draw (3.5,-2.5) [dotted,thick] ellipse (20pt and 10pt);
\draw (0.5,-2.5) [dotted,thick] ellipse (20pt and 10pt);
\draw (0.5,-4) [dotted,thick] ellipse (20pt and 10pt);
\draw (4.5,-4) [dotted,thick] ellipse (20pt and 10pt);
\draw (2.5,-4) [dotted,thick] ellipse (20pt and 10pt);
\draw (4.5,-5.5) [dotted,thick] ellipse (20pt and 10pt);
\draw (2.5,-5.5) [dotted,thick] ellipse (20pt and 10pt);
\end{tikzpicture}
\caption{Network and its line graph}
\label{fig:dual}
\end{center}
\end{figure}
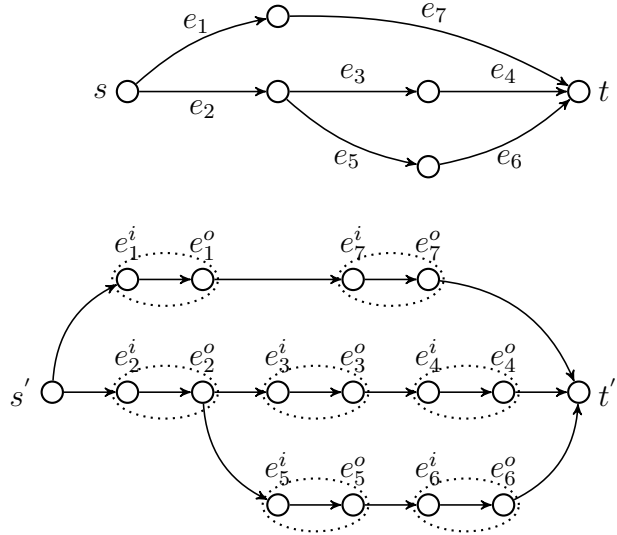






\end{document}